\documentclass{birkjour}
\usepackage{latexsym}
\usepackage{amsmath}
\usepackage{amsfonts}
\usepackage{amssymb}
\usepackage{amsthm}
\usepackage[noadjust]{cite}
 \newtheorem{thm}{Theorem}[section]

 \newtheorem{prop}[thm]{Proposition}
 \theoremstyle{definition}
 
 \theoremstyle{remark}

 \numberwithin{equation}{section}

\newcommand{\ed}{\end{document}}
\newcommand{\diag}{\mathrm{diag}}
\renewcommand{\sec}{\mathrm{sec}}
\newcommand{\End}{\mathrm{End}}
\newcommand{\Mat}{\mathrm{Mat}}
\newcommand{\id}{\mathrm{id}}

\newcommand{\LC}[3]{#1\,\rfloor_{#2}\,#3}  

\newcommand{\cl}{C \kern -0.1em \ell}
\newcommand{\hko}{\hookrightarrow}

\newcommand{\n}{\nonumber\\}

\newcommand{\bu}{\bullet}

\newcommand{\cv}{\circ}

\newcommand{\da}{\dagger}

\newcommand{\bec}{\begin{center}}
\newcommand{\eec}{\end{center}}

\newcommand{\bea}{\begin{array}}
\newcommand{\ear}{\end{array}}

\newcommand{\bfr}{\begin{flushright}}

\newcommand{\efr}{\end{flushright}}

\newcommand{\RR}{\mathbb{R}}\newcommand{\op}{\oplus}
\newcommand{\HH}{\mathbb{H}}

\newcommand{\bege}{\begin{equation}}
\newcommand{\enge}{\end{equation}}

\newcommand{\beq}{\begin{eqnarray}}\newcommand{\benu}{\begin{enumerate}}\newcommand{\enu}{\end{enumerate}}
\newcommand{\eeq}{\end{eqnarray}}

\newcommand{\vv}{{\bf v}}
\newcommand{\ee}{{\bf e}}
\newcommand{\uu}{{\bf u}}

\newcommand{\ww}{{\bf w}}

\newcommand{\ZZ}{\mathbb{Z}}

\newcommand{\x}{{\bf{X}}}

\newcommand{\OO}{\mathbb{O}}

\newcommand{\bx}{\begin{pmatrix}}

\def\pro#1{\!\Buildrel
	\raise 4pt\hbox{\the\scriptscriptfont0 #1}\under\circ\!}

\def\L{{\frak L}}

\def\O{\Bbb O}

\def\a{\alpha}
\def\b{\beta}
\def\x{\xi}

\def\y{\eta}

\def\d{\delta}
\def\da{\delta_\alpha}
\def\db{\delta_\beta}

\def\is{\!=\!}

\def\pro#1{\!\Buildrel
	\raise 4pt\hbox{\the\scriptscriptfont0 #1}\under\circ\!}

\begin{document}

%
%
%
%
%
%
%
%
%
\title[Non-Associativity in the Clifford Bundle on the 7-Sphere]
 {Non-Associativity in the Clifford Bundle\\ on the Parallelizable Torsion 7-Sphere}
\author{Rold\~ao da Rocha}
\address{%
Centro de Matem\'atica, Computa\c c\~ao e Cogni\c c\~ao\\
Universidade Federal do ABC\\
09210-170 Santo Andr\'e SP\\
Brazil}
\email{roldao.rocha@ufabc.edu.br}
\thanks{R. da Rocha is grateful to CNPq 476580/2010-2 and 304862/2009-6 for financial support, and to  Funda\c c\~ao de Amparo \`a Pesquisa do Estado de S\~ao Paulo (FAPESP) 2011/08710-1.
}
\author{M\'arcio A. Traesel}
\address{Instituto Federal de Educa\c c\~ao, Ci\^encia e Tecnologia de S\~ao Paulo\br
Campus Caraguatatuba\br
11665-310 Caraguatatuba SP\br
Brazil}
\email{marciotraesel@ifsp.edu.br}
\author{Jayme Vaz, Jr.}
\address{Departamento de Matem\'atica Aplicada\br
IMECC Unicamp CP 6065 13083-859\br
Campinas SP\br
Brazil}
\email{vaz@ime.unicamp.br}



\date{October 31, 2011}
\dedicatory{To Jaime Keller}

\begin{abstract}
In this paper we discuss generalized properties of non-associa\-ti\-vi\-ty in Clifford bundles on the 7-sphere $S^7$. Novel and prominent properties inherited from the non-associative structure of the Clifford bundle on $S^7$ are demonstrated. They naturally lead to general transformations of the spinor fields on $S^7$ and have dramatic consequences for the associated Ka\v{c}-Moody current algebras. All additional properties concerning the non-associative structure in the Clifford bundle on $S^7$ are considered. We further discuss and explore their applications.
\end{abstract}

\maketitle
\section{Introduction}
The main goal of this paper is provide a general class of non-associative structures on $S^7$ in the Clifford algebra formulation of generalized octonionic products. After briefly revisiting previous results \cite{eu1,qts7}, an octonionic product is defined in the Clifford algebra $\cl_{0,7}$ \cite{loun} which is closely related to the $S^7$ sphere \cite{aht}. Using this formalism, certain new identities are derived in a generalized octonionic algebra. See, for example, \cite{baez} for details as well as \cite{top,hofff,isot} for some prominent applications. Although there is a great variety of new octonionic products that may be defined, we restrict our formalism to a generalization of the results presented in \cite{dix, mart}. First, just to fix our notation, original non-associative deformed products between octonions presented in \cite{eu1,dix,mart} are briefly reviewed. Then, we discuss extended octonionic products between octonions and Clifford multivectors, and also among the Clifford multivectors, in the light of \cite{eu1,dix}. Our results immediately lead to the formalism presented in \cite{dix} in a very special case when a paravector component of an arbitrary multivector in $\cl_{0,7}$ is taken into account.

The results from \cite{eu1,qts7} are generalized along with a definition of non-equivalent non-associative products which are introduced in order to discuss the Clifford bundle on $S^7$ 
\cite{eu1}. The formalism presented in \cite{dix,mart} shows that the octonionic product can be deformed in order to include a parallelizable torsion on $S^7$. The $X$-product, as presented, is exactly twice the torsion components. We prove that, instead of considering the tangent space associated with the octonionic algebra, the whole Clifford bundle 
reveals unexpected properties. Despite of being hidden in the tangent bundle at $S^7$, the new properties are evinced when the Clifford bundle on $S^7$ is considered.

Section 2 briefly revisits  some mathematical tools and techniques related to the octonionic algebra in the Clifford algebra setting. We concentrate on the fundamental properties already 
discussed in \cite{eu1,loun} when introducing octonions via the Clifford algebra. In Section~3, new definitions reveal a wealth of unexpected results and the subtle difference arising in the generalization of the so called $u$-product with $u\in \cl_{0,7}$. Moreover, we review some properties from \cite{eu1} and present a few examples elucidating the main motivation for the formalism. In addition, new classes of non-associative products are introduced in the Clifford bundle on $S^7$ as well as directional non-associative products and new counter-examples to the Moufang identities that do not hold in our extended formalism.
 
\section{Preliminaries}
Let $V$ be an $n$-dimensional real vector space and let $V^*$ denote its dual. We consider the tensor algebra $\bigoplus_{i=0}^\infty T^i(V)$ from which we restrict our attention to the space 
$\bigwedge V  = \bigoplus_{k=0}^n\bigwedge^k V $ of multivectors over $V$. $\bigwedge^k V$ denotes the space of the antisymmetric $k$-tensors which is isomorphic to the vector space of $k$-forms.  Given 
$\psi\in\bigwedge^k V$, $\tilde\psi$ denotes the \emph{reversion}, an algebra anti-automorphism given by $\tilde{\psi} = (-1)^{[k/2]}\psi$ ([$k$] denotes the integer part of $k$). $\hat\psi$ denotes 
the \emph{main automorphism or grade involution}, given by $\hat{\psi} = (-1)^k \psi$. The \emph{conjugation} is defined as the reversion followed by the main automorphism. If $V$ is endowed with a non-degenerate, symmetric, bilinear map $g: V^*\times V^* \rightarrow \RR$, it is possible to extend $g$ to $\bigwedge V$.  The Clifford product between $\ww\in V$ and $\psi\in\bigwedge V$ is given by 
$\ww\psi = \ww\wedge \psi + \LC{\ww}{}{\psi}$. The Grassmann algebra $(\bigwedge V,g)$ endowed with the Clifford  product is denoted by $\cl(V,g)$ or $\cl_{p,q}$, the Clifford algebra associated with $V\simeq \RR^{p,q},\; p + q = n$.

The octonionic algebra via the Clifford algebras is briefly reviewed in \cite{eu1,loun,qts7}. The octonionic algebra $\mathbb{O}$ can be defined as the paravector space $\RR\op\RR^{0,7}$  endowed with the product $\circ$. The identity $\ee_0=1$ and an orthonormal basis $\left\{\ee_a\right\}^{7}_{a=1}$ generate the octonion algebra \cite{baez}. The octonionic product can be constructed using the Clifford algebra $\cl_{0, 7}$ as 
\bege\label{201}
A\circ B=\left\langle AB(1-\psi)\right\rangle_{0\op 1}, \quad A,B \in \RR \op \RR^{0, 7},
\enge
where $\psi=\ee_{126}+\ee_{237}+\ee_{341}+\ee_{452}+\ee_{563}+\ee_{674}+\ee_{715}\in\bigwedge^3 \RR^{0,7} \hookrightarrow \cl_{0,7}$.\footnote{We let $\ee_{ijk}$ denote the Clifford product 
$\ee_i\ee_j\ee_k$ where $\ee_i,\ee_j,\ee_k \in \mathbb{R}^{0,7}\hookrightarrow \cl_{0,7}$. In general, the Clifford product of $k$ vectors $\uu_{i_1}\uu_{i_2}\cdots \uu_{i_k}$ where 
$\uu_{i_s} \in \mathbb{R}^{0,7}$ will be denoted as $\uu_{i_1 \ldots i_k}$ \cite{loun}.} The main reason for introducing the octonionic product through the Clifford product as described is to extend  hereafter our formalism from the Clifford algebras to the Clifford bundles on $S^7$. The octonion multiplication table is constructed by $\ee_a\circ \ee_b=\epsilon^{c}_{ab}\ee_c-\delta_{ab}$ 
$(a,b,c=1,\ldots,7),$ where $\epsilon^{c}_{ab}=1$ for the cyclic permutations $(abc)=(126)$, $(237)$, $(341)$, $(452)$, $(563)$, $(674)$, $(715)$.  The Clifford conjugation of $X=X^0+X^a\ee_a \in \mathbb{O}$ is given by $\bar{X}=X^0-X^a\ee_a$, where $X^0$ and $X^a$ are real coefficients. 

\section{The $\bu$-product and the $\odot$-product on $S^7$} 
Given $X,Y \in \RR \op \RR^{0, 7}$ fixed but arbitrary such that $X\bar{X}=\bar{X}X=1=\bar{Y}Y=Y\bar{Y}$ ($X,Y \in S^7$), the \textit{$X$-product} is defined by \cite{mart, dix} 
\bege\label{205}
A\circ_X B:=(A\circ X)\circ(\bar{X}\circ B).
\enge
The expressions below are shown in, e.g., \cite{dix}
\bege\label{3p14}
(A\circ X)\circ(\bar{X}\circ B)=X\circ ((\bar{X}\circ A)\circ B)=(A\circ(B\circ X))\circ \bar{X}.
\enge
These identities are in general valid for any octonion $X,$ not only for $X \in S^7$. Since from now on we focus on $S^7$ with parallelizable torsion, we restrict the expressions above to 
$X\in S^7$. The structure of the vector space $\mathbb{R}^{0,7} \hookrightarrow \cl_{0,7}$ is not sufficient to determine whether the $\OO$-conjugation is the grade involution or conjugation from inherited from the Clifford algebra $\cl_{0,7}$ through the equations above since $\bar{X}$ (the octonionic conjugation) could either be $\hat{X}$ (the grade involution) or $\bar{X}$ (the Clifford conjugation). Since $\bar{X}$ (octonionic) is an anti-automorphism, it excludes immediately the grade involution.

Because of the non-associativity of the product in $\OO$, in general, $A\cv_X B\neq A\cv B$. Exceptions can be provided when one defines \cite{e8} the following sets:
\begin{align*}
\Xi_{0} = & \{\pm \ee_{a}\}, \\
\Xi_{1} = & \{(\pm\ee_{a}\pm \ee_{b})/\sqrt{2}\,|\, a,b \mbox{ distinct}\}, \\ 
\Xi_{2} = & \{(\pm \ee_{a}\pm \ee_{b}\pm \ee_{c}\pm \ee_{d})/2\,|\, a,b,c,d \mbox{ distinct}, \; \ee_{a}\circ(\ee_{b}\circ(\ee_{c}\circ\ee_{d}))=\pm 1\}, \\ 
\Xi_{3} = & \{(\sum_{a=0}^{7}\pm \ee_{a})/\sqrt{8}\,|\, \mbox{ odd number of plus signs}  \},
\end{align*} 
and when we choose $X$ in one these four sets. Then, for all $a,b,c,d\in\{0,\ldots,7\}$, there is some $c\in\{0,\ldots,7\}$ such that $\ee_{a}\circ_{X}\ee_{b} = \pm \ee_{c}.$ Eqs.~(\ref{3p14}) show that $X$ determines two linear transformations $f, g\in$ End($\OO$) such that
\begin{gather}
A \circ_X B = f(A \circ f^{-1}(B)) = g(g^{-1}(A) \circ B)
\end{gather}
for all $A, B \in \OO$. The alternativity of the $\circ_X$-multiplication then follows as
\begin{gather}
A\circ_X (A\circ_X B)=(A\circ A)\circ_X B,  \qquad (A\circ_X B)\circ_X B=A\circ_X (B\circ B).
\end{gather}
The $\circ_X$-multiplication is essentially the initial $\circ$-multiplication: there exists an orthogonal transformation $h$ in the special orthogonal group SO($\RR^{0,7}$) such that the mapping $\lambda + \vv \mapsto\lambda + h(\vv)$, for all $\lambda\in \RR$ and $\vv \in \RR^{0,7}$, is an isomorphism from $(\RR\oplus\RR^{0,7}, \circ)=\OO$ onto $(\RR\oplus\RR^{0,7}, \circ_X) = \OO_X$. The reciprocal statement might be conjectured: if there is a $\star$-multiplication $\OO \times \OO \rightarrow \OO$ and a transformation $h \in$ SO($\RR^{0,7}$) such that the mapping $\lambda + \vv \mapsto\lambda + h(\vv)$ is an isomorphism from $\OO$ onto $\OO_X$, then there exists $X \in S^7$ such that $A\star B = A\circ_X B$ for all $A, B \in \OO$. The isomorphism $\OO\simeq \OO_X$ was widely 
discussed in \cite{dix}. 

Now, a simple change in the sign of $X$ gives rise to distinct but isomorphic copies of $\OO$. Thus, an orbit containing the isomorphic copies of $\OO$ arising from any given copy, is 
$ S^7/\ZZ_2=\RR \mathrm{P}^7,$ the manifold obtained from the $S^7$-sphere by identifying two diametrically opposite points since $A\cv_{-X}B=A\cv_X B$. Moreover, the composition of $X$-products is yet another $X$-product. That is, for $X,Y \in S^7$,
\begin{align} 
AB\xrightarrow{X} A\cv_X B &= (A\cv X)\cv(\bar{X}\cv B)\xrightarrow{Y}(A\cv_X Y)\cv_X(\bar{Y}\cv_X B) \notag\\ 
                           &= \left[A\cv(Y\cv X)\right]\cv\left[(\overline{Y\cv X})\cv B\right]=A\cv_{Y\cv X}B
\end{align}
using the fact that $(U\cv\bar{X})\cv X=U=\bar{X}\cv(X\cv U)$ for all $U \in \OO$. 

A non-associative product called the \textit{$u$-product} was introduced in \cite{eu1} as a natural generalization for the $X$-product. For homogeneous multivectors 
$u=\uu_{1 \ldots k} = \uu_{1} \cdots \uu_{k}\in\bigwedge^k \RR^{0,7}\hookrightarrow\cl_{0, 7}$, where $\left\{\uu_p\right\}^{k}_{p=1}\subset\RR^{0, 7}$  is an orthogonal set with respect to the metric $g=\diag(-,-,-,-,-,-,-)$, and  $A\in\RR\op\RR^{0, 7}$, the products $\bu_\llcorner$ and $\bu_\lrcorner$ are defined as \cite{eu1}
\begin{align}
\bu_\llcorner \colon &(\RR \op \RR^{0,7})\times \bigwedge^k \RR^{0,7} \to \RR \op \RR^{0, 7},\n 
&(A,u)\mapsto A\bu_\llcorner u=((\cdots((A\circ \uu_1)\circ \uu_2)\circ \cdots)\circ \uu_{k-1})\circ \uu_k, \label{209}\\
\bu_\lrcorner \colon &\bigwedge^k \RR^{0,7}\times (\RR \op \RR^{0,7})\to \RR \op \RR^{0,7},\n
&(u,A)\mapsto u\bu_\lrcorner A=\uu_1\circ (\uu_2\cv(\cdots \circ (\uu_{k-1}\circ (\uu_k \circ A))\cdots )).\label{210}
\end{align}
In addition, one defines $A\bu_\llcorner(\lambda\,1)=\lambda A =(\lambda\,1)\bu_\lrcorner A$ for any $A\in\OO$, a real scalar $\lambda,$ and where $1=\ee_0$ denotes the unity of $\cl_{0,7}.$ Moreover, the products (\ref{209}), (\ref{210}) are extended to the whole $\bigwedge \RR^{0,7}$ by linearity. Given an  element $u \in \bigwedge \RR^{0,7}$, the $u$-product is defined as
\begin{align}
\cv_u \colon (\RR \op \RR^{0,7})\times (\RR \op \RR^{0,7}) &\to \RR \op \RR^{0,7} \n 
(A, B)&\mapsto A\cv_u B:=(A\bu_\llcorner u)\cv (\bar{u}\bu_\lrcorner B).
\label{213}
\end{align}
In \cite{eu1}, the authors ask whether 
\begin{multline}
A\cv_u B:=(A\bu_\llcorner u)\cv (\bar{u}\bu_\lrcorner B)\overset{?}{=}\\(A\cv(B\bu_\llcorner u))\bu_\llcorner \bar{u}\overset{?}{=}u\bu_\lrcorner((\bar{u}\bu_\lrcorner A)\cv B)
\end{multline}
holds in a context where any similar generalization related to~\eqref{205} can be constructed in the non-associative formalism induced by the $u$-product for 
$u\in  \sec \left(\bigwedge T_xS^7\right)$, the exterior bundle on $S^7$. In a straightforward example we showed in \cite{qts7} that, in general, 
\begin{gather}
(A\bu_\llcorner u)\cv (\bar{u}\bu_\lrcorner B)\neq (A\cv(B\bu_\llcorner u))\bu_\llcorner \bar{u}.
\end{gather} 
Whether the equality 
\begin{gather}
(A\circ X)\circ(\bar{X}\circ B)=(A\circ(B\circ X))\circ \bar{X}
\end{gather} 
holds in the more general setting  in the context of the $\bu$-product, is an open question which we intend to address. 

Now, given 
$$
X = X^{0}+X^{1}\ee_{1}+X^{2}\ee_{2}+X^{3}\ee_{3}+X^{4}\ee_{4}+X^{5}\ee_{5}+X^{6}\ee_{6}+X^{7}\ee_{7}\in S^{7},
$$ 
then
\begin{align} 
\ee_{1}\circ_{X}\ee_{2} =& ((X^{0})^{2}+(X^{1})^{2}+(X^{2})^{2}+(X^{6})^{2} \notag \\
                         &\hspace*{0.5in}-(X^{3})^{2}-(X^{4})^{2}-(X^{5})^{2}-(X^{7})^{2})\ee_{6} \notag \\
                         &+2(X^{0}X^{5}+X^{1}X^{7}-X^{2}X^{4}+X^{3}X^{6})\ee_{3} \notag \\ 
                         &+2(-X^{0}X^{7}+X^{1}X^{5}+X^{2}X^{3}+X^{4}X^{6})\ee_{4} \notag \\
                         &+2(-X^{0}X^{3}-X^{1}X^{4}-X^{2}X^{7}+X^{5}X^{6})\ee_{5} \notag \\ 
                         &+2(X^{0}X^{4}-X^{1}X^{3}+X^{2}X^{5}+X^{7}X^{6})\ee_{7}. 
\label{ex1}
\end{align}
The Hopf fibration $S^3\cdots S^7\rightarrow S^4$ can therefore be defined as \cite{dix}: 
\begin{align} 
A^3=&2(X^{0}X^{5}+X^{1}X^{7}-X^{2}X^{4}+X^{3}X^{6}),\notag \\
A^4=&2(-X^{0}X^{7}+X^{1}X^{5}+X^{2}X^{3}+X^{4}X^{6}),\notag \\
A^5=&2(-X^{0}X^{3}-X^{1}X^{4}-X^{2}X^{7}+X^{5}X^{6}),\notag \\
A^6=&((X^{0})^{2}+(X^{1})^{2}+(X^{2})^{2}+(X^{6})^{2}\notag\\
    &\hspace*{0.5in}-(X^{3})^{2}-(X^{4})^{2}-(X^{5})^{2}-(X^{7})^{2}),\notag \\
A^7=&2(X^{0}X^{4}-X^{1}X^{3}+X^{2}X^{5}+X^{7}X^{6}).
\end{align}
The equation~(\ref{ex1}) can be written as
\begin{gather} 
\ee_1\cv_X\ee_2=A^3\ee_3+A^4\ee_4+A^5\ee_5+A^6\ee_6+A^7\ee_7
\end{gather}
where $A \in \OO_X$ and $A \in S^4$. One can observe that the $X$-product is a map from $S^7$ to $S^4$.

$\cl_{0,7}$ is a Clifford algebra whose matrix representation is provided by the direct sum of two ideals -- each isomorphic to $\Mat(8,\RR)$ -- which are generated, respectively, by the central idempotents $\frac12 (1 \pm \ee_{1234567}).$ Let us consider two linear mappings  $\RR^{0,7}\rightarrow \End(\OO)$  which map every vector $\vv\in\RR^{0,7}$ to the linear operator 
$L_\vv :\OO\rightarrow\OO$ and $R_\vv:\OO\rightarrow\OO,$ respectively, such that $L_\vv(A) = \vv\circ A$ and $R_\vv(A) = A\circ \vv$ for all $A\in\OO.$ Given $C\in\OO$, from the identities 
$(A\circ A)\circ C=A\circ (A\circ C)$ and $(C\circ A)\circ A=C\circ (A\circ A),$ $C\in\OO$, it follows that $L_\vv\circ L_\vv = R_\vv \circ R_\vv = \vv^2 \id_\OO$ (here, $\circ$ denotes the composition of mappings). According to the universal property of the Clifford algebra, the mapping $\vv\mapsto L_\vv$ sending $\vv\mapsto R_\vv$ extends to an algebra [anti-]morphism 
$\cl_{0,7}\rightarrow \End(\OO).$ Furthermore, one can verify that 
\begin{gather}
(1+\ee_{1234567})\bullet_\lrcorner A=A\bullet_\llcorner (1+\ee_{1234567})=0, \quad \forall A\in \mathbb{O},
\end{gather} 
and since $\dim(\cl_{0,7}) = 128$ and $\dim(\End(\OO)) = 64,$ the kernel of the morphism and of the antimorphism is the ideal generated by one of the idempotents $\frac12(1 \pm \ee_{1234567}).$ The equality
$$
(((((\ee_1 \circ \ee_2)\circ \ee_3)\circ \ee_4)\circ \ee_5)\circ \ee_6)\circ \ee_7 = \ee_1 \circ (\ee_2 \circ(\ee_3 \circ (\ee_4 \circ(\ee_5 \circ (\ee_6 \circ \ee_7))))) = -1
$$
shows that in both cases the kernel is the ideal generated by the central element $\frac12(1 + \ee_{1234567}).$

The operators $L_u$ and $R_u$ are defined for all $u \in \cl_{0,7}$, and if we set $L_u(A) = u\bullet_\llcorner A$ and $R_u(A) = A \bullet_\lrcorner u$, then for $A,B\in\OO$ and $u,w\in\cl_{0,7}$, one obtains
\begin{gather}
A \bullet_\llcorner B = A \bullet_\lrcorner B = A\circ  B, \notag\\
(uw) \bullet_\llcorner B = u \bullet_\llcorner (w \bullet_\llcorner B), \quad B \bullet_\lrcorner (uw) = (B \bullet_\lrcorner u) \bullet_\lrcorner w,
\end{gather} 
which immediately implies Eqs.~(\ref{209}, \ref{210}).

Now, given $u=\uu_{1 \ldots k}$ and $v=\vv_{1 \ldots k} \in \cl_{0,7}$, the non-associative product between Clifford algebra elements was defined in~\cite{eu1} as
\begin{align}
\odot_\llcorner \colon &\cl_{0,7} \times \cl_{0,7}\to \RR \op \RR^{0,7},\n
&(u, v)\mapsto u\odot_\llcorner v:=\uu_1 \cv (\uu_2 \cv(\cdots \cv(\uu_{k-1}\cv(\uu_k \bu_\llcorner v))\cdots)),\n
\odot_\lrcorner \colon &\cl_{0, 7} \times \cl_{0, 7} \to \RR\op \RR^{0,7},\n
&(u, v)\mapsto u\odot_\lrcorner v:=((\cdots\cv((u\bu_\lrcorner \vv_1)\cv \vv_2)\cv\cdots)\cv \vv_{k-1})\cv \vv_k.
\label{221}
\end{align}
Symbol $\odot$ denotes both products $\odot_\llcorner$ and $\odot_\lrcorner.$ It is easy to see that when elements of $\cl_{0, 7}$ are restricted to the paravector space $\RR \op \RR^{0,7}$, then 
$A\bu B\equiv A\cv B$ and $A \odot B \equiv A \cv B$ where $A,B \in \RR \op \RR^{0,7}.$ It was seen in \cite{qts7} that 
$(\ee_{57}-\ee_{31})\odot_\llcorner\ee_{123}=\ee_7+\ee_2$, while $(\ee_{57}-\ee_{31})\odot_\lrcorner\ee_{123}=-\ee_7+\ee_2.$

\section{Towards a Moufang-like generalization}
Our goal is to obtain the most general expression generalizing~\eqref{3p14} and emulating it when $u\in\cl_{0,7}$ is considered instead of $X\in\RR\oplus\RR^{0,7}.$

In order to generalize~\eqref{205}, some results must be introduced first. 
\begin{prop}[\cite{dix}] 
The following elements from $\bigwedge^0 \RR^{0,7}\op\bigwedge^3 \RR^{0,7}$
\begin{align} 
P_0&=(1+\ee_{476}+\ee_{517}+\ee_{621}+\ee_{732}+\ee_{143}+\ee_{254}+\ee_{365})/8,\n
P_1&=(1-\ee_{476}+\ee_{517}+\ee_{621}-\ee_{732}+\ee_{143}-\ee_{254}-\ee_{365})/8,\n
P_2&=(1-\ee_{476}-\ee_{517}+\ee_{621}+\ee_{732}-\ee_{143}+\ee_{254}-\ee_{365})/8,\n
P_3&=(1-\ee_{476}-\ee_{517}-\ee_{621}+\ee_{732}+\ee_{143}-\ee_{254}+\ee_{365})/8,\n
P_4&=(1+\ee_{476}-\ee_{517}-\ee_{621}-\ee_{732}+\ee_{143}+\ee_{254}-\ee_{365})/8,\n
P_5&=(1-\ee_{476}+\ee_{517}-\ee_{621}-\ee_{732}-\ee_{143}+\ee_{254}+\ee_{365})/8,\n
P_6&=(1+\ee_{476}-\ee_{517}+\ee_{621}-\ee_{732}-\ee_{143}-\ee_{254}+\ee_{365})/8,\n
P_7&=(1+\ee_{476}+\ee_{517}-\ee_{621}+\ee_{732}-\ee_{143}-\ee_{254}-\ee_{365})/8,
\label{222}
\end{align}
are $\bu$-idempotents. 
\end{prop}
Indeed, it is straightforward to verify that for all $a,b\in\{0,\ldots,7\}$ and for all $A\in\OO$, the relations
\begin{align}
P_b \bullet_\lrcorner (P_a \bullet_\lrcorner A) &= \delta_{ab}P_a \bullet_\lrcorner A\label{bu2},\\
(A \bullet_\llcorner P_a) \bullet_\llcorner P_b &= \delta_{ab}\,A\bullet_\llcorner P_a\label{bu3}, 
\end{align} 
hold  where $\delta_{ab}=1$ if $a=b$, and 0 otherwise. The relations (\ref{bu2}) and (\ref{bu3}) are understood as  $\bu$-actions on the octonions or an $\odot$-action upon Clifford algebra elements $\cl_{0,7}$, as defined in~(\ref{209}), (\ref{210}), and (\ref{221}). In this direction, the $P_a$'s are $\bu$-idempotents satisfying $\sum^{7}_{a=0}P_a=1$. The idempotents 
$\left\{P_a\right\}^{7}_{a=0}$ form a complete set of \textit{orthogonal} $\bu$-idempotents decomposing the unity~\cite{dix}.

Given a fixed but arbitrary $a \in \left\{1,2,\ldots,7\right\}$, consider the idempotent $P_a$ in~\eqref{222}. Then, $P_a$ is a linear combination of $3$-vectors in $\bigwedge^3 \RR^{0,7}$ such that 
the element $\ee_{ijk}$ appears in the combination with the plus sign only when one of the indices $i,j,k$ equals $a.$
Furthermore, define 
\begin{equation}\label{alfa}
\alpha_a=2P_a-1 \in \bigwedge^0 \RR^{0,7}\op\bigwedge^3 \RR^{0,7}.
\end{equation}
\begin{prop} 
Let $\alpha_a$ be as in~\eqref{alfa} with $a \in \{1, \ldots, 7\}.$ Then, for all $X \in \mathbb{O}$,
\begin{equation}
\alpha_a\bu_\lrcorner(X\cv \ee_a)=\bar{X}\cv\ee_a.
\label{lemanovo}
\end{equation}
\end{prop}
\begin{proof} 
From the definition of $\alpha_a$ in~\eqref{alfa} and  $\bu_\lrcorner$ in~\eqref{213}, we have
\begin{align*}
\alpha_a\bu_\lrcorner(X\cv \ee_a)&=(2P_a-1)\bu_\lrcorner(X\cv\ee_a)\\
&=(2P_a-1)\bu_\lrcorner((X^0\ee_a + \sum_{i=1}^{7 }X^i\ee_i)\cv\ee_a)\\
&=2P_a\bu_\lrcorner(X^0\ee_a+ \sum_{i=i}^{7} X^i (\ee_i\cv\ee_a)) -X^0\ee_a - \sum_{i=1}^{7}  X^i (\ee_i\cv\ee_a)\\
&=2X^0\ee_a-X^0\ee_a- \sum_{i=1}^{7} X^i(\ee_i\cv\ee_a) =\bar{X}\cv\ee_a,
\end{align*}
for every $1\leq a \leq 7.$
\end{proof}
A simpler expression for the $P_a$ is  $\ee_a(1-\psi)\ee_a^{-1}/8$ (for $a = 0, 1, \ldots, 7$) where $\psi$ is given by Eq.~(\ref{201}) and its subsequent line. When $\OO$ is regarded as a left or right module over the Clifford algebra, the elements $P_a$ are mutually annihilating idempotents modulo the ideal $(1 + \ee_{1234567}) \cl_{0,7}\hko \cl_{0,7}.$ This assertion can be verified by a direct calculation using twenty one formulas like this one:
\begin{gather}
\ee_{476} \ee_{517}- \ee_{732} = \ee_{517} \ee_{476} - \ee_{732} = -(1 + \ee_{1234567}) \ee_{732}.
\end{gather}
Then, we also have these relations: 
\begin{gather}
P_a\bullet_\lrcorner \ee_a=\ee_a\bullet_\llcorner P_a=\ee_a \quad \mbox{but} \quad P_a\bullet_\lrcorner \ee_b=\ee_b\bullet_\llcorner P_a=0 \quad \mbox{if} \quad a\neq b.
\end{gather}

Recall that $\tilde{u}$ denotes the reversion of $u$ in $\bigwedge \mathbb{R}^{0,7}.$
\begin{prop} 
\begin{enumerate}
\item[(a)] Given $\ee_0 \in \mathbb{O}$ and $u=\ee_{i_1}\ee_{i_2}\cdots \ee_{i_k} \in \bigwedge^k \RR^{0,7}$, $k=1,2,\ldots,6$, then
\begin{gather}
\ee_0\bu_\llcorner u=\rho\ee_0\cv(1\bu_\llcorner \tilde{u})\nonumber
\end{gather}
where $\rho=(-1)^{|u|(|u|-1)/2}$ and $|u|$ denotes the degree of $u;$ if $u\in\bigwedge^k \RR^{0,7}$ then $|u|$ = $k.$
\item[(b)] Given $\ee_a \in \mathbb{O}$ and $u=\ee_{i_1}\ee_{i_2}\cdots \ee_{i_k} \in \bigwedge^k \RR^{0,7}$, $k=1,2,\ldots,6$, with $\ee_a\notin\{\ee_{i_1}, \ee_{i_2},\ldots, \ee_{i_k}\}$ or 
$\{\ee_a, \ee_l, \ee_m\}$ not being an $\HH$-triple for $\ee_l\neq \ee_m$ and $\ee_l, \ee_m\in \{\ee_{i_1},\ee_{i_2},\ldots,\ee_{i_6}\}$, then
\begin{gather}
\ee_a\bu_\llcorner u=\lambda\ee_a\cv(1\bu_\llcorner \tilde{u})
\end{gather}
where $\lambda=(-1)^{(|u|^2+|u|+2)/2}.$
\end{enumerate}
\end{prop}
\begin{proof} 
\begin{enumerate}
\item[(a)]\begin{enumerate}
\item[0)] For $u=\ee_0\in\bigwedge^0 \RR^{0,7}$, then $\ee_0\bu_\llcorner u=\ee_0\bu_\llcorner\ee_0=\ee_0\cv(1\cv\ee_0)=\ee_0\cv(1\bu_\llcorner\tilde{u}).$
\item[1)] For $u=\ee_b\in\bigwedge^1 \RR^{0,7}$: 
\bege
\ee_0\bu_\llcorner u=\ee_0\bu_\llcorner\ee_b = \ee_0\cv (1 \cv \ee_b)= \ee_0\cv (1\bu_\llcorner \tilde{u}).
\nonumber
\enge
\item[2)] For $u = \ee_{bc}\in\bigwedge^2 \RR^{0,7}$: 
\bege
\ee_0\bu_\llcorner u=\ee_0\bu_\llcorner \ee_{bc} = (\ee_0\cv \ee_b)\cv \ee_c = \ee_0\cv (\ee_b\cv \ee_c)=\ee_0\cv (1\bu_\llcorner \ee_{bc})=-\ee_0\cv (1\bu_\llcorner \tilde{u}).
\nonumber
\enge
\item[3)] For $u = \ee_{bcd}\in\bigwedge^3 \RR^{0,7}$:
\begin{align}
\ee_0\bu_\llcorner u &= \ee_0\bu_\llcorner \ee_{bcd} = ((\ee_0\cv \ee_b)\cv \ee_c)\cv \ee_d = (\ee_0\cv (\ee_b\cv \ee_c))\cv \ee_d\n 
&=\ee_0\cv ((\ee_b\cv \ee_c)\cv \ee_d) = \ee_0\cv (1\bu_\llcorner\ee_{bcd})= -\ee_0\cv (1\bu_\llcorner \tilde{u}).
\nonumber
\end{align}
\item[4)] For $u = \ee_{bcdf}\in\bigwedge^4 \RR^{0,7}$:
\begin{align}
\ee_0\bu_\llcorner u &= \ee_0\bu_\llcorner \ee_{bcdf} = (((\ee_0\cv \ee_b)\cv \ee_c)\cv \ee_d)\cv \ee_f = ((\ee_0\cv (\ee_b\cv \ee_c))\cv \ee_d)\cv \ee_f \n 
&=(\ee_0\cv ((\ee_b\cv \ee_c)\cv \ee_d))\cv \ee_f = \ee_0\cv (((\ee_b\cv \ee_c)\cv \ee_d)\cv \ee_f)\n
&=\ee_0\cv (1\bu_\llcorner\ee_{bcdf})=\ee_0 \cv (1\bu_\llcorner \tilde{u}).
\nonumber
\end{align}
\item[5)] For $u = \ee_{bcdfg}\in\bigwedge^5 \RR^{0,7}$:
\begin{align}
\ee_0\bu_\llcorner u &= \ee_0\bu_\llcorner \ee_{bcdfg} = ((((\ee_0\cv \ee_b)\cv \ee_c)\cv \ee_d)\cv \ee_f)\cv \ee_g\n
&=(((\ee_0\cv (\ee_b\cv \ee_c))\cv \ee_d)\cv \ee_f)\cv \ee_g=((\ee_0\cv ((\ee_b\cv \ee_c)\cv \ee_d))\cv \ee_f)\cv \ee_g\n
&=(\ee_0\cv (((\ee_b\cv \ee_c)\cv \ee_d)\cv \ee_f))\cv \ee_g=\ee_0\cv ((((\ee_b\cv \ee_c)\cv \ee_d)\cv \ee_f)\cv \ee_g)\n
&=\ee_0\cv(1\bu_\llcorner\ee_{bcdfg})=\ee_0 \cv (1\bu_\llcorner \tilde{u}).
\nonumber
\end{align}
\item[6)] For $u = \ee_{bcdfgh}\in\bigwedge^6 \RR^{0,7}$:
\begin{align}
\ee_0\bu_\llcorner u &= \ee_0\bu_\llcorner \ee_{bcdfgh} = (((((\ee_0\cv \ee_b)\cv \ee_c)\cv \ee_d)\cv \ee_f)\cv \ee_g)\cv \ee_h \n
&=((((\ee_0\cv (\ee_b\cv \ee_c))\cv \ee_d)\cv \ee_f)\cv \ee_g)\cv \ee_h\n
&=(((\ee_0\cv ((\ee_b\cv \ee_c)\cv \ee_d))\cv \ee_f)\cv \ee_g)\cv \ee_h\n 
&=((\ee_0\cv (((\ee_b\cv \ee_c)\cv \ee_d)\cv \ee_f))\cv \ee_g)\cv \ee_h\n
&=(\ee_0\cv ((((\ee_b\cv \ee_c)\cv \ee_d)\cv \ee_f)\cv \ee_g))\cv \ee_h\n
&=\ee_0\cv (((((\ee_b\cv \ee_c)\cv \ee_d)\cv \ee_f)\cv \ee_g)\cv \ee_h)\n
&=\ee_0\cv(1\bu_\llcorner\ee_{bcdfgh})=-\ee_0 \cv (1\bu_\llcorner \tilde{u}).
\nonumber
\end{align}
\item[7)] For $u = \ee_{bcdfghj}\in\bigwedge^7 \RR^{0,7}$:
\begin{align}
\ee_0\bu_\llcorner u &= \ee_0\bu_\llcorner \ee_{bcdfghj} = ((((((\ee_0\cv \ee_b)\cv \ee_c)\cv \ee_d)\cv \ee_f)\cv \ee_g)\cv \ee_h)\cv \ee_j\n
&=(((((\ee_0\cv (\ee_b\cv \ee_c))\cv \ee_d)\cv \ee_f)\cv \ee_g)\cv \ee_h)\cv \ee_j\n
&=((((\ee_0\cv ((\ee_b\cv \ee_c)\cv \ee_d))\cv \ee_f)\cv \ee_g)\cv \ee_h)\cv \ee_j\n
&=(((\ee_0\cv (((\ee_b\cv \ee_c)\cv \ee_d)\cv \ee_f))\cv \ee_g)\cv \ee_h)\cv \ee_j\n
&=((\ee_0\cv ((((\ee_b\cv \ee_c)\cv \ee_d)\cv \ee_f)\cv \ee_g))\cv \ee_h)\cv \ee_j\n
&=(\ee_0\cv (((((\ee_b\cv \ee_c)\cv \ee_d)\cv \ee_f)\cv \ee_g)\cv \ee_h))\cv \ee_j\n
&=\ee_0\cv ((((((\ee_b\cv \ee_c)\cv \ee_d)\cv \ee_f)\cv \ee_g)\cv \ee_h)\cv \ee_j)\n
&=\ee_0\cv(1\bu_\llcorner\ee_{bcdfghj})=-\ee_0 \cv (1\bu_\llcorner \tilde{u}).
\nonumber
\end{align}
Therefore, $\ee_0\bu_\llcorner u=\rho\ee_0\cv(1\bu_\llcorner\tilde{u})$, where $\rho=(-1)^{|u|(|u|-1)/2}$ which matches the reversion sign.
\end{enumerate}
\item[(b)]\begin{enumerate}
\item[0)] When $u=\ee_0\in\bigwedge^0 \RR^{0,7}$, it follows that:
\bege
\ee_a\bu_\llcorner\ee_0=\ee_a\cv\ee_0=\ee_a\cv(1\cv\ee_0)=\ee_a\cv(1\bu_\llcorner \tilde{u}).
\nonumber
\enge
\item[1)] When $u=\ee_b\in\bigwedge^1 \RR^{0,7}$, it follows that:
\begin{enumerate}
\item[(i)] $a\neq b$: $\ee_a\bu_\llcorner u = \ee_a\bu_\llcorner\ee_b = \ee_a\cv (1 \cv \ee_b)= \ee_a\cv (1\bu_\llcorner \tilde{u})$,
\item[(ii)] $a=b$: $\ee_a\bu_\llcorner u = \ee_a\bu_\llcorner\ee_b = \ee_a\cv\ee_a= \ee_a\cv (1 \cv \ee_a)=\ee_a\cv (1\bu_\llcorner \tilde{u}).$
\end{enumerate}
\item[2)] When $u = \ee_{bc}\in\bigwedge^2 \RR^{0,7}$, it follows that:
\begin{enumerate}
\item[(i)] $a\notin \{b, c\}$, and ($abc$) is not an $\HH$-triple:  
\bege
\ee_a\bu_\llcorner u = \ee_a\bu_\llcorner \ee_{bc} = (\ee_a\cv \ee_b)\cv \ee_c = -\ee_a\cv (\ee_b\cv \ee_c)=-\ee_a\cv (1\bu_\llcorner \ee_{bc})=\ee_a\cv (1\bu_\llcorner \tilde{u}),
\nonumber
\enge
\item[(ii)] $a\in \{b, c\}$ or ($abc$) is an $\HH$-triple. Without a loss of generality, consider $a=b$: 
\begin{align}
\ee_a\bu_\llcorner u&=\ee_a\bu_\llcorner\ee_{bc}=\ee_a\bu_\llcorner \ee_{ac} = (\ee_a\cv \ee_a)\cv \ee_c = \ee_a\cv (\ee_a\cv \ee_c)\n
&=\ee_a\cv (1\bu_\llcorner \ee_{ac})=-\ee_a\cv (1\bu_\llcorner \tilde{u}).\nonumber
\end{align}
\end{enumerate}
\item[3)] When $u = \ee_{bcd}\in\bigwedge^3 \RR^{0,7}$, it follows that:
\begin{enumerate}
\item[(i)] $a\notin \{b, c, d\}$, and ($ijk$) is not an $\HH$-triple where $i,j,k$ \newline $\in \{a,b,c,d\}$: 
\begin{align}
\ee_a\bu_\llcorner u &= \ee_a\bu_\llcorner \ee_{bcd} = ((\ee_a\cv \ee_b)\cv \ee_c)\cv \ee_d =  -(\ee_a\cv (\ee_b\cv \ee_c))\cv \ee_d\n 
&=\ee_a\cv ((\ee_b\cv \ee_c)\cv \ee_d) = \ee_a\cv (1\bu_\llcorner\ee_{bcd})= -\ee_a\cv (1\bu_\llcorner \tilde{u}),
\nonumber
\end{align}
\item[(ii)] $a\in \{b, c, d\}$ or ($abc$) is an $\HH$-triple. Without a loss of generality, consider $a=b$: 
\begin{align}
\ee_a\bu_\llcorner u &= \ee_a\bu_\llcorner \ee_{bcd} = \ee_a\bu_\llcorner\ee_{acd}=((\ee_a\cv \ee_a)\cv \ee_c)\cv \ee_d = (\ee_a\cv (\ee_a\cv \ee_c))\cv \ee_d\n
&=-\ee_a\cv ((\ee_a\cv \ee_c)\cv \ee_d)=-\ee_a\cv (1\bu_\llcorner\ee_{acd})=\ee_a\cv (1\bu_\llcorner \tilde{u}),
\nonumber
\end{align}
\item[(iii)] ($ijk$) is an $\HH$-triple where $i,j,k\in \{a,b, c, d\}$ and ($ijk)\neq(abc$). Let us suppose that ($abd$) is an $\HH$-triple: 
\begin{align}
\ee_a\bu_\llcorner u &=\ee_a\bu_\llcorner \ee_{bcd} = ((\ee_a\cv \ee_b)\cv \ee_c)\cv \ee_d = -(\ee_a\cv (\ee_b\cv \ee_c))\cv \ee_d \n 
&=\ee_a\cv ((\ee_b\cv \ee_c)\cv \ee_d)=\ee_a(1\bu_\llcorner\ee_{bcd})=-\ee_a\cv (1\bu_\llcorner \tilde{u}).
\nonumber
\end{align}
\end{enumerate}
\item[4)]  When $u = \ee_{bcdf}\in\bigwedge^4 \RR^{0,7}$, it follows that:
\begin{enumerate}
\item[(i)] $a\notin \{b, c, d, f\}$, and ($ijk$) is not an $\HH$-triple where $i,j,k \in \{a,b,c,d,f\}$: 
\begin{align}
\ee_a\bu_\llcorner u &=\ee_a\bu_\llcorner \ee_{bcdf} = (((\ee_a\cv \ee_b)\cv \ee_c)\cv \ee_d)\cv \ee_f\n
&= -((\ee_a\cv (\ee_b\cv \ee_c))\cv \ee_d)\cv \ee_f=(\ee_a\cv ((\ee_b\cv \ee_c)\cv \ee_d))\cv \ee_f\n
&= -\ee_a\cv (((\ee_b\cv \ee_c)\cv \ee_d)\cv \ee_f)= -\ee_a\cv (1\bu_\llcorner\ee_{bcdf})=-\ee_a \cv (1\bu_\llcorner \tilde{u}),
\nonumber
\end{align}
\item[(ii)] $a\in \{b, c, d, f\}$ or ($abc$) is an $\HH$-triple. Without a loss of generality, consider $a=b$: 
\begin{align*}
\ee_a\bu_\llcorner u &= \ee_a\bu_\llcorner \ee_{bcdf}=\ee_a\bu_\llcorner\ee_{acdf} = (((\ee_a\cv \ee_a)\cv \ee_c)\cv \ee_d)\cv \ee_f\n
&=((\ee_a\cv (\ee_a\cv \ee_c))\cv \ee_d)\cv \ee_f=-(\ee_a\cv ((\ee_a\cv \ee_c)\cv \ee_d))\cv \ee_f\n
&=\ee_a\cv (((\ee_a\cv \ee_c)\cv \ee_d)\cv \ee_f)=\ee_a\cv(1\bu_\llcorner\ee_{acdf})=\ee_a \cv (1\bu_\llcorner \tilde{u}),
\end{align*}
\item[(iii)] ($ijk$) is an $\HH$-triple where $i,j,k\in \{a,b, c, d, f\}$ and ($ijk)\neq(abc$). Let us suppose that ($abd$) is an $\HH$-triple:
\begin{align}
\ee_a\bu_\llcorner u &= \ee_a\bu_\llcorner \ee_{bcdf} = (((\ee_a\cv \ee_b)\cv \ee_c)\cv \ee_d)\cv \ee_f\n
&= -((\ee_a\cv (\ee_b\cv \ee_c))\cv \ee_d)\cv \ee_f= (\ee_a\cv ((\ee_b\cv \ee_c)\cv \ee_d))\cv \ee_f\n
&= -\ee_a\cv (((\ee_b\cv \ee_c)\cv \ee_d)\cv \ee_f)=-\ee_a\cv(1\bu_\llcorner\ee_{bcdf})=-\ee_a \cv (1\bu_\llcorner \tilde{u}).
\nonumber
\end{align}
\end{enumerate}
\item[5)] When $u = \ee_{bcdfg}\in\bigwedge^5 \RR^{0,7}$, it follows that:
\begin{enumerate}
\item[(i)] $a\notin \{b, c, d, f, g\}$, and ($ijk$) is not an $\HH$-triple where $i,j,k \in \{a,b,c,d,f,g\}$:
\begin{align}
&\ee_a\bu_\llcorner u = \ee_a\bu_\llcorner \ee_{bcdfg} = ((((\ee_a\cv \ee_b)\cv \ee_c)\cv \ee_d)\cv \ee_f)\cv \ee_g\n
&=-(((\ee_a\cv (\ee_b\cv \ee_c))\cv \ee_d)\cv \ee_f)\cv \ee_g=((\ee_a\cv ((\ee_b\cv \ee_c)\cv \ee_d))\cv \ee_f)\cv \ee_g\n
&=-(\ee_a\cv (((\ee_b\cv \ee_c)\cv \ee_d)\cv \ee_f))\cv \ee_g=\ee_a\cv ((((\ee_b\cv \ee_c)\cv \ee_d)\cv \ee_f)\cv \ee_g)\n
&=\ee_a\cv(1\bu_\llcorner\ee_{bcdfg})=\ee_a \cv (1\bu_\llcorner \tilde{u}),
\nonumber
\end{align}
\item[(ii)] $a\in \{b, c, d, f, g\}$ or ($abc$) is an $\HH$-triple. Without a loss of generality, consider $a=b$: 
\begin{align}
&\ee_a\bu_\llcorner u = \ee_a\bu_\llcorner \ee_{bcdfgh}=\ee_a\bu_\llcorner\ee_{acdfgh} = ((((\ee_a\cv \ee_a)\cv \ee_c)\cv \ee_d)\cv \ee_f)\cv \ee_g\n
&=(((\ee_a\cv (\ee_a\cv \ee_c))\cv \ee_d)\cv \ee_f)\cv \ee_g=-((\ee_a\cv ((\ee_a\cv \ee_c)\cv \ee_d))\cv \ee_f)\cv \ee_g\n 
&=(\ee_a\cv (((\ee_a\cv \ee_c)\cv \ee_d)\cv \ee_f))\cv \ee_g=-\ee_a\cv ((((\ee_a\cv \ee_c)\cv \ee_d)\cv \ee_f)\cv \ee_g)\n
&=-\ee_a\cv(1\bu_\llcorner\ee_{acdfgh})=-\ee_a \cv (1\bu_\llcorner \tilde{u}),
\nonumber
\end{align}
\item[(iii)] ($ijk$) is an $\HH$-triple where $i,j,k\in \{a,b, c, d, f, g\}$ and ($ijk)\neq(abc$). Let us suppose that ($abd$) is an $\HH$-triple:
\begin{align}
&\ee_a\bu_\llcorner u= \ee_a\bu_\llcorner \ee_{bcdfg} = ((((\ee_a\cv \ee_b)\cv \ee_c)\cv \ee_d)\cv \ee_f)\cv \ee_g\n
&=-(((\ee_a\cv (\ee_b\cv \ee_c))\cv \ee_d)\cv \ee_f)\cv \ee_g=((\ee_a\cv ((\ee_b\cv \ee_c)\cv \ee_d))\cv \ee_f)\cv \ee_g\n
&=-(\ee_a\cv (((\ee_b\cv \ee_c)\cv \ee_d)\cv \ee_f))\cv \ee_g=\ee_a\cv ((((\ee_b\cv \ee_c)\cv \ee_d)\cv \ee_f)\cv \ee_g)\n
&=\ee_a\cv(1\bu_\llcorner\ee_{bcdfg})=\ee_a \cv (1\bu_\llcorner \tilde{u}),
\nonumber
\end{align}
\item[(iv)] ($ijk$) and $(lmn$) are $\HH$-triples where $i,j,k,l,m,n$ \newline $\in \{a,b,c,d,f,g\}$, $\{i,j,k\}\neq\{l,m,n\}$ and ($ijk)\neq(abc$). Let us suppose that ($abd$) and ($cfg$) are 
$\HH$-triples: 
\begin{align}
&\ee_a\bu_\llcorner u = \ee_a\bu_\llcorner \ee_{bcdfg}=((((\ee_a\cv \ee_b)\cv \ee_c)\cv \ee_d)\cv \ee_f)\cv \ee_g\n
&=-(((\ee_a\cv (\ee_b\cv \ee_c))\cv \ee_d)\cv \ee_f)\cv \ee_g=((\ee_a\cv ((\ee_b\cv \ee_c)\cv \ee_d))\cv \ee_f)\cv \ee_g\n
&=-(\ee_a\cv (((\ee_b\cv \ee_c)\cv \ee_d)\cv \ee_f))\cv \ee_g=\ee_a\cv ((((\ee_b\cv \ee_c)\cv \ee_d)\cv \ee_f)\cv \ee_g)\n
&=\ee_a\cv(1\bu_\llcorner\ee_{bcdfg})=\ee_a \cv (1\bu_\llcorner \tilde{u}).
\nonumber
\end{align}
\end{enumerate}
\item[6)] When $u = \ee_{bcdfgh}\in\bigwedge^6 \RR^{0,7}$, it follows that:
\begin{enumerate}
\item[(i)] $a\notin \{b, c, d, f, g, h\}$, and ($ijk$) and $(lmn$) are not $\HH$-triples where $i,j,k,l,m,n \in \{a,b,c,d,f,g,h\}$ and $\{i,j,k\}\neq\{l,m,n\}$:
\begin{align}
\ee_a\bu_\llcorner u &= \ee_a\bu_\llcorner \ee_{bcdfgh}=(((((\ee_a\cv \ee_b)\cv \ee_c)\cv \ee_d)\cv \ee_f)\cv \ee_g)\cv \ee_h \n
&=-((((\ee_a\cv (\ee_b\cv \ee_c))\cv \ee_d)\cv \ee_f)\cv \ee_g)\cv \ee_h\n
&=(((\ee_a\cv ((\ee_b\cv \ee_c)\cv \ee_d))\cv \ee_f)\cv \ee_g)\cv \ee_h\n 
&=-((\ee_a\cv (((\ee_b\cv \ee_c)\cv \ee_d)\cv \ee_f))\cv \ee_g)\cv \ee_h\n
&=(\ee_a\cv ((((\ee_b\cv \ee_c)\cv \ee_d)\cv \ee_f)\cv \ee_g))\cv \ee_h\n
&=-\ee_a\cv (((((\ee_b\cv \ee_c)\cv \ee_d)\cv \ee_f)\cv \ee_g)\cv \ee_h)\n
&=-\ee_a\cv(1\bu_\llcorner\ee_{bcdfgh})=\ee_a \cv (1\bu_\llcorner \tilde{u}),
\nonumber
\end{align}
\item[(ii)] $a\in \{b, c, d, f, g, h\}$ or ($abc$) is an $\HH$-triple. Without a loss of generality, consider $a=b$: 
\begin{align}
\ee_a\bu_\llcorner u &= \ee_a\bu_\llcorner \ee_{bcdfgh}=\ee_a\bu_\llcorner\ee_{acdfgh}\n
&= (((((\ee_a\cv \ee_a)\cv \ee_c)\cv \ee_d)\cv \ee_f)\cv \ee_g)\cv \ee_h \n
&=((((\ee_a\cv (\ee_a\cv \ee_c))\cv \ee_d)\cv \ee_f)\cv \ee_g)\cv \ee_h\n
&=-(((\ee_a\cv ((\ee_a\cv \ee_c)\cv \ee_d))\cv \ee_f)\cv \ee_g)\cv \ee_h\n 
&=((\ee_a\cv (((\ee_a\cv \ee_c)\cv \ee_d)\cv \ee_f))\cv \ee_g)\cv \ee_h\n
&=-(\ee_a\cv ((((\ee_a\cv \ee_c)\cv \ee_d)\cv \ee_f)\cv \ee_g))\cv \ee_h\n
&=\ee_a\cv (((((\ee_a\cv \ee_c)\cv \ee_d)\cv \ee_f)\cv \ee_g)\cv \ee_h)\n
&=\ee_a\cv(1\bu_\llcorner\ee_{acdfgh})=-\ee_a \cv (1\bu_\llcorner \tilde{u}),
\nonumber
\end{align}
\item[(iii)] ($ijk$) is an $\HH$-triple where $i,j,k\in \{a,b, c, d, f, g, h\}$ and ($ijk)\neq(abc$). Let us suppose that ($abd$) is an $\HH$-triple:
\begin{align}
\ee_a\bu_\llcorner u &= \ee_a\bu_\llcorner \ee_{bcdfgh} = (((((\ee_a\cv \ee_b)\cv \ee_c)\cv \ee_d)\cv \ee_f)\cv \ee_g)\cv \ee_h \n
&=-((((\ee_a\cv (\ee_b\cv \ee_c))\cv \ee_d)\cv \ee_f)\cv \ee_g)\cv \ee_h\n
&=(((\ee_a\cv ((\ee_b\cv \ee_c)\cv \ee_d))\cv \ee_f)\cv \ee_g)\cv \ee_h\n 
&=-((\ee_a\cv (((\ee_b\cv \ee_c)\cv \ee_d)\cv \ee_f))\cv \ee_g)\cv \ee_h\n
&=(\ee_a\cv ((((\ee_b\cv \ee_c)\cv \ee_d)\cv \ee_f)\cv \ee_g))\cv \ee_h\n
&=-\ee_a\cv (((((\ee_b\cv \ee_c)\cv \ee_d)\cv \ee_f)\cv \ee_g)\cv \ee_h)\n
&=-\ee_a\cv(1\bu_\llcorner\ee_{bcdfgh})=\ee_a \cv (1\bu_\llcorner \tilde{u}),
\nonumber
\end{align}
\item[(iv)] ($ijk$) and $(lmn$) are $\HH$-triples where $i,j,k,l,m,n$ \newline $\in \{a,b,c,d,f,g,h\}$, $\{i,j,k\}\neq\{l,m,n\}$ and ($ijk)\neq(abc$). Let us suppose that ($abd$) and ($cfg$) are 
$\HH$-triples:
\begin{align}
\ee_a\bu_\llcorner u &= \ee_a\bu_\llcorner \ee_{bcdfgh} = (((((\ee_a\cv \ee_b)\cv \ee_c)\cv \ee_d)\cv \ee_f)\cv \ee_g)\cv \ee_h \n
&=((((\ee_a\cv (\ee_b\cv \ee_c))\cv \ee_d)\cv \ee_f)\cv \ee_g)\cv \ee_h\n
&=-(((\ee_a\cv ((\ee_b\cv \ee_c)\cv \ee_d))\cv \ee_f)\cv \ee_g)\cv \ee_h\n 
&=((\ee_a\cv (((\ee_b\cv \ee_c)\cv \ee_d)\cv \ee_f))\cv \ee_g)\cv \ee_h\n
&=-(\ee_a\cv ((((\ee_b\cv \ee_c)\cv \ee_d)\cv \ee_f)\cv \ee_g))\cv \ee_h\n
&=\ee_a\cv (((((\ee_b\cv \ee_c)\cv \ee_d)\cv \ee_f)\cv \ee_g)\cv \ee_h)\n
&=\ee_a\cv(1\bu_\llcorner\ee_{bcdfgh})=-\ee_a \cv (1\bu_\llcorner \tilde{u}).
\nonumber
\end{align}
\end{enumerate}
\end{enumerate}
\end{enumerate}
\end{proof}
This complements Propositions 1, $1^\prime$, 2, $2^\prime$, 3, and 4 in \cite{qts7} towards the required generalization of~(\ref{3p14}). Indeed, using the aforementioned and demonstrated results we can show that
\begin{gather}
(\ee_a \bullet_\lrcorner u)\circ(\bar{u}\bullet_\llcorner \ee_b)=\beta (\ee_a \circ (\ee_b\bullet_\llcorner u))\circ(1\bullet_\llcorner \tilde{u})
\label{euue}
\end{gather} 
for $a\neq b$ while $\beta=\pm 1$ which depends on the degree $k$ of $u\in\bigwedge^k\mathbb{R}^{0,7}.$ When $u$ is an octonion, identity~(\ref{euue}) leads to~(\ref{3p14}). However, until now a general expression conjectured to be 
\begin{gather}
(A\bu_\llcorner u)\cv (\bar{u}\bu_\lrcorner B) = [\star_{u}^1({A})\cv ({\star_{u}^2}({B})\bu_\llcorner u)]\cv \overline{(1\bullet_\lrcorner \tilde{u})},
\end{gather}
where $\star_{u}^1$ and $\star_{u}^2$  are  $u$-induced involutions on $\mathbb{O}$ distinct from the $\mathbb{O}$-conju\-ga\-tion, lacks.

A result similar to the last Proposition can be obtained when $u\in\cl_{0,7}$ and $\vv\in\RR^{0,7}$. Then, the identities (\ref{205}) and (\ref{3p14}) imply
\begin{gather}
\vv\bu_\llcorner u = (1\bu_\llcorner (\vv u \vv^{-1}))\circ \vv. 
\label{2205}
\end{gather} 
If $\vv\in \{\ee_a\}_{a=1}^7$ and if $u$ is a product of $k$ pairwise orthogonal vectors in the basis $\{\ee_a\}_{a=1}^7$ of $\RR^{0,7}$, then $\vv^{-1} = -\vv$ and $\vv u\vv^{-1} = su$, for some 
$s=\pm 1$, and $\vv\circ (1\bu_\llcorner u)\circ \vv^{-1} =t(1\bu_\llcorner u)$ for some $t=\pm 1$, hence
\begin{gather}
\ee_a \bu_\llcorner u = st \ee_a \circ (1\bu_\llcorner u).
\label{2206}
\end{gather}
It remains to calculate the factor $st$ in (\ref{2206}). The factor $s$ is $(-1)^{k-1}$ or $(-1)^k$ provided $\ee_a$ is, or is not, a factor in the product $u.$  The factor $t$ is 1 when 
$1\bu_\llcorner u$ equals $\pm 1$ or $\pm \ee_a$, and it is $-1$ in all other cases. To know which into case $1\bu_\llcorner u$ falls, we can calculate it up to the sign  by means of an Abelian group consisting of the eight sets $\{\pm\ee_b\}_{b=1}^7.$

\section{Concluding remarks and outlook}
The parameter $A\circ_X B = (A\circ X)\circ(\bar{X}\circ B) = \bar{X}\circ((X\circ A)\circ B)$ is twice the parallelizing torsion whose components are given by 
\begin{gather}
T_{ijk}(X)=[(\bar{\ee}_i\circ\bar{X})\circ(X\circ {\ee_j})]\circ \ee_k.
\label{torsion}
\end{gather}
The right-hand-side of~(\ref{torsion}) is exactly the $\bar{X}$-product between $\bar{\ee}_i$ and $\ee_j$. So, the $S^7$ algebra can be written as $[\delta_i,\delta_j]=2T_{ijk}(X)\delta_k$ where 
$\delta_AX = X\cv A$, and the variation $\delta$ denotes the parallelizing covariant derivative \cite{dix}. By means of the $\odot$-product, all subsequent octonionic products are regarded as the 
$\odot$-product involving the Clifford multivector associated with the given octonionic product as defined in~(\ref{209}), (\ref{210}), and (\ref{221}). Thus, the arbitrary number of octonionic products can be encoded in a unique product -- the $\odot$-product.
 It is not quite a straightforward task to consider the reversed non-associative products. For instance, given 
$\alpha_0$ in~(\ref{alfa}), the identity 
\begin{multline}
(\ee_a\cv \ee_b)\cv X =(\alpha_0 \bu_\lrcorner(X\cv{\ee_b}))\cv \ee_a -\\(\alpha_0 \bu_\lrcorner(X\cv\bar{\ee}_a))\cv \bar{\ee}_b + (X\cv {\ee_b})\cv \ee_a
\end{multline}
holds for all $X\in\mathbb{O}.$ The possibility of performing non-associative products between arbitrary multivectors of $\cl_{0,7}$ naturally arises in our formalism \cite{eu1}, and it generalizes furthermore the formalism introduced in~\cite{dix} concerning the original $X$-product. Some additional application are shown in~\cite{hofff}. 

The formalism developed here is one more step towards new applications of the $S^7$ spinor fields. Heretofore, the product $\circ_u$ was introduced and now a few words delving into novel applications. It is well-known that the tangent space at $X$ is spanned by the units $\{X\circ \ee_i\}_{i\is 1}^7.$ As in \cite{ree}, by considering the tangent space basis at two infinitesimally separated points,  the parallel transport of this basis is defined by an infinitesimal transformation $\delta_A X=X\circ A$ where $A$ is a pure octonion with no scalar part. Given a field $\x$ on~$S^7$, the commutator of two such transformations can be calculated explicitly \cite{mart} as:
\begin{align*}
[\da,\db]\x &\equiv \da(\db \x)-\db(\da \x)=(\x\circ\a)\circ\b-(\x\circ\b)\circ\a\\
	    &=\x\circ\bigl(\bar{X}\circ((X\circ\a)\circ\b)-\bar{X}\circ((X\circ\b)\circ\a)\bigr)\\
            &=\d_{\bar{X}\circ((X\circ\a)\circ\b)-\bar{X}\circ((X\circ\b)\circ\a)}\circ\x
\end{align*}
The parameter 
\begin{gather}
\bar{X}\circ((X\circ\a)\circ\b)-\bar{X}\circ((X\circ\b)\circ\a)= 2(\bar{X}\circ((X\circ\a)\circ\b))
\end{gather} 
is twice the parallelizing torsion \cite{roo}. In component notation, 
\begin{gather}
T_{abc}(X)=[(\bar{\ee}_a\circ\bar{X})(X\circ \ee_b)\circ \ee_c] \quad  \mbox{and} \quad  [\d_a,\d_b]=2T_{abc}(X)\d_c.
\end{gather}
The variation $\d$ is indeed the parallelizing covariant derivative. We want to introduce a boson $\y$ such that ${\y\over |\y|}\is Y$ with some $S^7$ transformation rule. This excludes the simplest candidate $\da Y\is Y \circ \a$ \cite{mart}. The two fields are bound to transform differently. The correct transformation rule turns out to be $\delta_A \eta=\eta\circ_X A.$ Such transformation is  related to the transformation of the  parameter field $X.$ Therefore, fermions cannot transform without the presence of a parameter field. A field (bosonic or fermionic) transforming according to such map is  a  spinor under $S^7.$ Our formalism introduces a new transformation of such spinor fields since the product $\circ_u$ requires more parameters than the $X$-product. The current algebra associated to such spinor fields is, in addition, dramatically modified. We postpone  a deeper discussion of these consequences to a forthcoming paper since it is far beyond the scope of the present work. Even though a huge variety of new products can be introduced using our formalism, we are concerned to reveal and descry some applications. The products here introduced are immediate generalization of the results in, e.g., Cederwall, Bengtsson, Rooman, Preitschopf, Brink \cite{dix, mart, roo}, as well as other ones obtained by Toppan, G{\"u}naydin, Lukierski, Ketov, de Wit, Nicolai, Gursey, and others \cite{top, toplu, Guna2}. Finally, objects described here provide immediate generalizations of the instanton Hopf fibration and Lounesto spinor field classification \cite{hofff} as well as generalizations of Clifford algebras~\cite{esta} and the Lounesto spinor field classification in eight dimensions~\cite{where}.

\subsection*{Acknowledgment}
R. da Rocha is grateful to CNPq 476580/2010-2 and 304862/2009-6, and to  Funda\c c\~ao de Amparo \`a Pesquisa do Estado de S\~ao Paulo (FAPESP) 2011/08710-1, for financial support.

\end{document}